\documentclass[journal]{IEEEtran}
\usepackage{cite}

\usepackage{color,soul}
\usepackage{gensymb}
\usepackage{dsfont}

\ifCLASSINFOpdf
  \usepackage[pdftex]{graphicx}
  \graphicspath{{../pdf/}{../jpeg/}}
  \DeclareGraphicsExtensions{.pdf,.jpeg,.png}
\else
  \usepackage[dvips]{graphicx}
  \graphicspath{{../eps/}}
  \DeclareGraphicsExtensions{.eps}
\fi
\usepackage[cmex10]{amsmath}
\usepackage{amssymb}
\DeclareMathOperator*{\argmaxA}{arg\,max}

\usepackage{euscript}

\usepackage{multirow}
\usepackage{algorithm}
\usepackage[]{algpseudocode}
\usepackage{diagbox}
\usepackage{physics}
\usepackage{tikz}

\usepackage{amsthm}
\newtheorem{definition}{Definition}
\newtheorem{theorem}{Theorem}

\ifCLASSOPTIONcompsoc
  \usepackage[caption=false,font=normalsize,labelfont=sf,textfont=sf]{subfig}
\else
  \usepackage[caption=false,font=footnotesize]{subfig}
\fi
\usepackage{url}

\begin{document}
%
\title{Outage Detection in Partially Observable Distribution Systems using Smart Meters and Generative Adversarial Networks}
%
%
%

\author{Yuxuan~Yuan,~\IEEEmembership{Student Member,~IEEE,}
	Kaveh~Dehghanpour,~\IEEEmembership{Member,~IEEE,}
	Fankun~Bu,~\IEEEmembership{Student Member,~IEEE,}
	and Zhaoyu~Wang,~\IEEEmembership{Member,~IEEE}
\thanks{This work is supported by the U.S. Department of Energy Office of Electricity, and National Science Foundation under ECCS 1929975 (\textit{Corresponding author: Zhaoyu Wang.})

 Y. Yuan, K. Dehghanpour, F. Bu, and Z. Wang are with the Department of
Electrical and Computer Engineering, Iowa State University, Ames,
IA 50011 USA (e-mail: yuanyx@iastate.edu; wzy@iastate.edu).
 }
}
%
%

\markboth{Submitted to IEEE for possible publication. Copyright may be transferred without notice}%
{Shell \MakeLowercase{\textit{et al.}}: Bare Demo of IEEEtran.cls for Journals}

%



\maketitle

\begin{abstract}
In this paper, we present a novel data-driven approach to detect outage events in partially observable distribution systems by capturing the changes in smart meters' (SMs) data distribution. To achieve this, first, a breadth-first search (BFS)-based mechanism is proposed to decompose the network into a set of \textit{zones} that maximize outage location information in partially observable systems. Then, using SM data in each zone, a generative adversarial network (GAN) is designed to implicitly extract the temporal-spatial behavior in \textit{normal conditions} in an unsupervised fashion. After training, an anomaly scoring technique is leveraged to determine if real-time measurements indicate an outage event in the zone. Finally, to infer the location of the outage events in a multi-zone network, a zone coordination process is proposed to take into account the interdependencies of intersecting zones. We have provided analytical guarantees of performance for our algorithm using the concept of \textit{entropy}, which is leveraged to quantify outage location information in multi-zone grids. The proposed method has been tested and verified on distribution feeder models with real SM data.
\end{abstract}


\section{Introduction}\label{introduction}
Outage detection is a challenging problem in power systems, especially in distribution networks where the majority of outage events take place. According to the statistical data provided by the U.S Energy Information Administration (EIA), each customer lost power for around 4 hours on average in 2016 \cite{EIAoutage}. To decrease outage duration, and improve system reliability and customer satisfaction, distribution system operators (DSOs) deploy state-of-the-art outage management systems (OMS), using modern software tools and protection devices with bidirectional communication function. This allows DSOs to collect real-time up-to-the-second data from the network \cite{RA2018}. Nevertheless, use of intelligent communication-capable devices in distribution systems has not become prevalent, mostly due to budgetary limitations of utilities \cite{ST2018}. Hence, identification of distribution system outage events, especially for small utilities, still relies on trouble calls from customers and manual inspection. However, trouble calls alone are not a reliable data source of outage detection because customers may not make prompt calls to utilities \cite{HSF2016}. Also, conventional expert-experience-based outage discovery methods that use customer calls are laborious, costly, and time-consuming \cite{FCL2014}. 

In recent years, a number of papers have explored data-driven alternatives for outage detection. According to the type of data source, the previous works in this area can be classified into two groups: \textit{Class I - Smart meter (SM)-based methods}: With the widespread deployment of advanced metering infrastructure (AMI), SMs provide an opportunity to rapidly detect outage events by recording the real-time demand consumption and automatically sending ``last gasp" signals to the utilities. In \cite{ZSH2018}, a multi-label support vector machine classification method is presented that utilizes the last gasp signals of SMs to detect and find the locations of damaged lines in fully observable networks. In \cite{RM2018}, a hierarchical framework is developed to provide anomaly-related insights using multivariate event counter data collected from SMs. In \cite{SJC2015}, a fuzzy Petri nets-based approach is proposed to detect nontechnical losses and outage events by tracking the differences between profiled and irregular power consumption. In \cite{KS2001}, a probabilistic and fuzzy model-based algorithm is presented to process outage data using AMI. In \cite{RA2001}, a tree-based polling algorithm is developed to obtain information about the system conditions by polling local SMs. \textit{Class II - non-SM-based methods}: Other data sources have been used in the literature for outage detection, as well. In \cite{RA2018}, a hypothesis testing-based outage detection method is developed combining the use of real-time power flow measurements and load forecasts of the nodes. In \cite{HSF2016}, a social network-based data-driven method is proposed by leveraging real-time information extraction from Twitter. In \cite{PK2014}, a new boosting algorithm is developed to estimate outages in overhead distribution systems by utilizing weather information.

Even though previous works provide valuable results, critical questions remain unanswered in this area. The limitation of most Class I models is their basic assumption that the distribution system is \textit{fully observable}, i.e., all the nodes have measurement devices. However, this assumption does not necessarily apply to practical systems, in which large portions of customers do not own smart meters \cite{ZSH2018}. On the other hand, Class II methods are generally based on several limiting assumptions, such as availability of accurate forecasts for customer loads, availability of real-time power flow measurements, and reliability of social network data. Another difficulty in outage detection is \textit{outage data scarcity}, which means that the size of the outage data is far smaller compared to the data in normal conditions. This issue causes a \textit{data imbalance problem} that could hinder reliable training of supervised learning-based outage detection models \cite{VC2007}.

To address these shortcomings, in this paper, a generative adversarial network (GAN)-based method is developed to detect power outages in partially observable distribution systems by capturing the anomalous changes in SMs' measurement data distributions that are caused by outage events \cite{GAN2014}. Compared to the previous works, the proposed method solves three fundamental challenges in outage monitoring for partially observable distribution systems: 1) Unlike supervised classifiers that can fail in case of outage data scarcity, the proposed generative model follows an \textit{unsupervised learning} style which only relies on the operation data in normal conditions for model training. Then, a GAN-based anomaly score is defined to quantify the deviations between the learned distribution and the real-time measurements to detect potential outage events, i.e. new observations with high anomaly scores imply outage \cite{GAN2017}. 2) Due to the temporal variability of AMI data, efficient outage detection requires capturing high-dimensional temporal-spatial relationships in measurement data. Conventional data distribution estimators are limited by the high-dimensional nature of the data. Instead of constructing a complex data likelihood function explicitly, our approach trains GANs to implicitly extract the underlying distribution of the data. Each GAN consists of two interconnected deep neural networks (DNNs) \cite{AC2018}. 3) Considering the partial observability of real systems, we have proposed a breadth-first search (BFS)-based mechanism to decompose large-scale distribution networks into a set of intersecting \textit{zones} \cite{BFS}. Each zone is determined by two neighboring observable nodes of the network (i.e. nodes with known voltages and demands) and contains only a subset of network branches. A separate GAN is trained in each zone using the time-series data of the two observable nodes. Since sectionalizing networks into multiple zones can be done in more than one way depending on the choice of observable nodes, it is necessary to find the optimal set of zones. Our BFS-based approach optimizes the zone selection and anomaly score coordination process and achieves maximum outage location information. To demonstrate this, we have proposed an outage detection metric based on the information-theoretic concept of \textit{entropy} to quantify outage location information. The proposed outage detection methodology has been tested and verified using real AMI data and network models.



\section{Real Data Description and Zone Selection}\label{framework}
\subsection{AMI Data Description}
The available AMI historical data used in this paper contains several U.S. mid-west utilities' hourly energy consumption data (kWh) and voltage magnitude measurements of over 6000 customers \cite{website}. The dataset includes around four years of measurements, from January 2015 to May 2018. Over 95\% of customers are residential and commercial loads in the dataset. The hourly data was initially processed to remove bad and missing data caused by communication error.

\subsection{Outage Detection Zone Definition}

When an outage happens in a radial system, a protective device isolates the faulted area along with the loads downstream of the fault location \cite{RA2018}. This will cause the measurement data samples from unfaulted upstream observable nodes to deviate from the data distribution in normal condition. In this paper, we exploit this phenomenon to define an outage detection zone. 
\begin{figure}[tbp]
	\centering
	\includegraphics[width=3.5in]{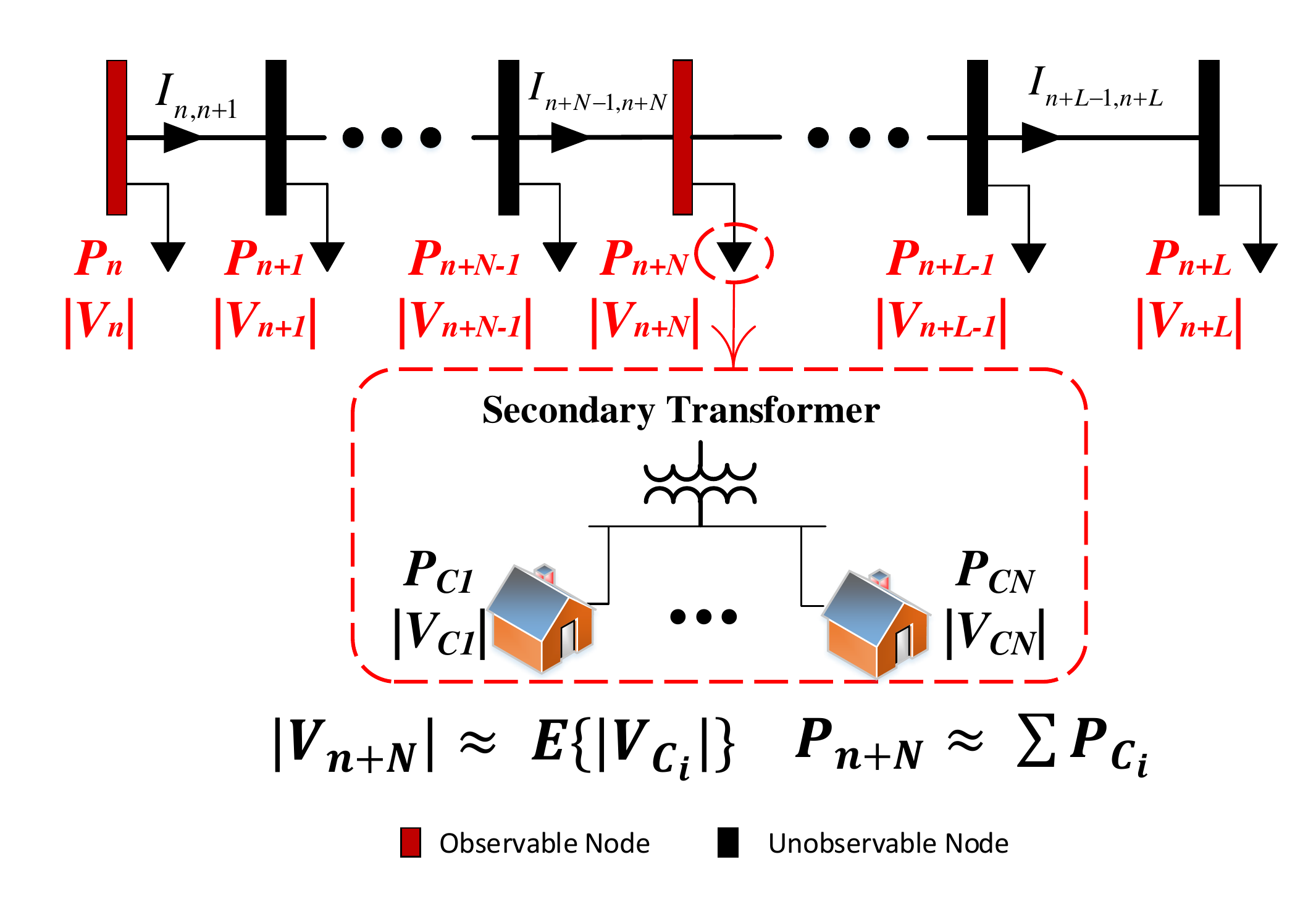}
	\caption{Example zone in normal condition.}
	\label{fig:outage_inside}
		\vspace{-1em}
\end{figure}
\begin{figure}[tbp]
	\centering
	\includegraphics[width=3.5in]{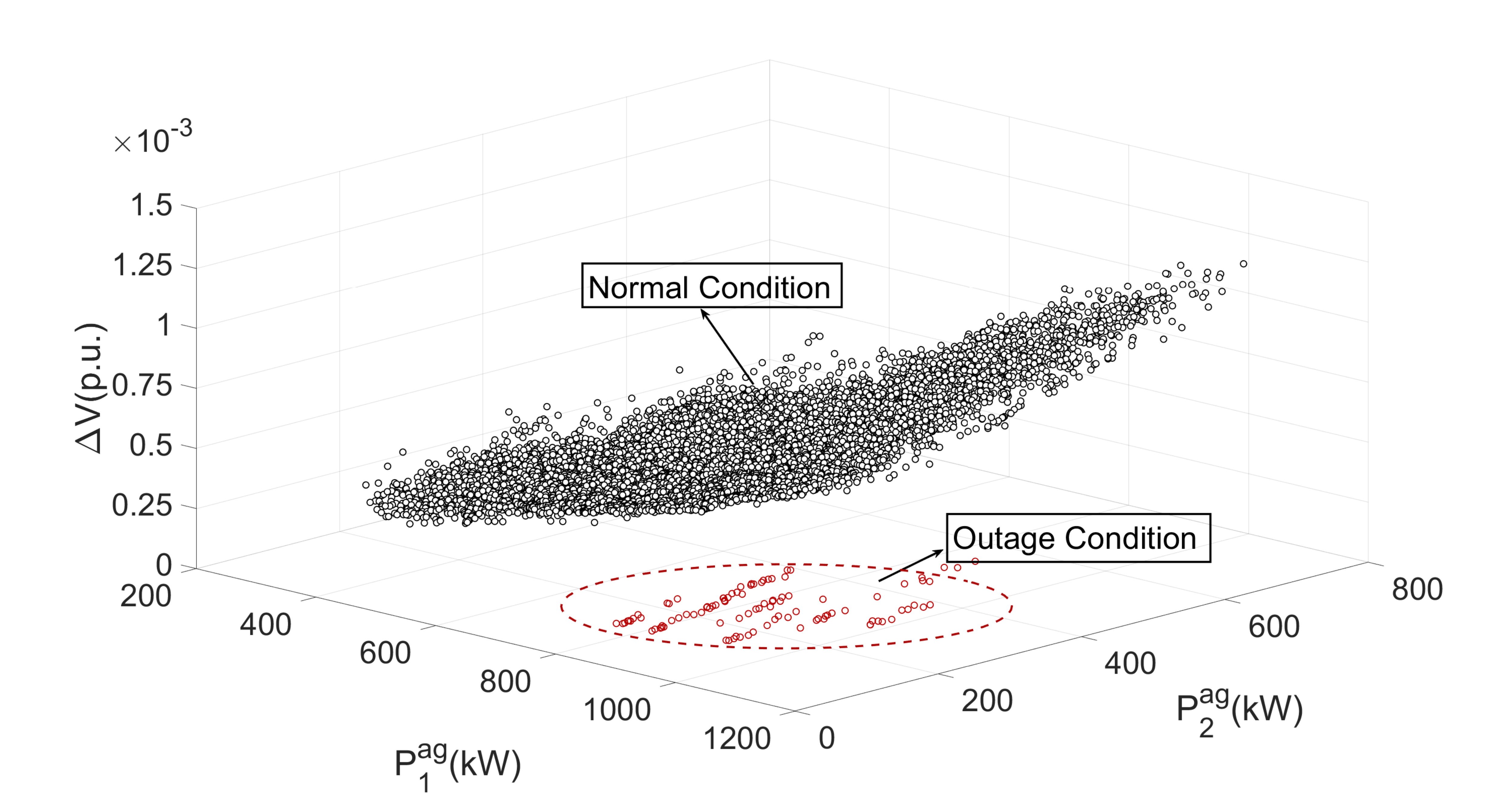}
	\caption{Joint data distribution under normal and outage conditions.}
	\label{fig:3d}
		\vspace{-1em}
\end{figure}
In general, two observable nodes (i.e. nodes with AMI-based measured voltage magnitudes and power consumption) on the same path can be utilized to define an outage detection zone. To show this, Fig. \ref{fig:outage_inside} presents a typical distribution feeder with two observable nodes, node $n$ and node $n+N$. Given the radial structure of the feeder, the voltage drop, $\Delta{V}$, between nodes $n$ and $n+N$ can be expressed as \cite{KW2016}:
\begin{equation}
\label{eq:normal}
\Delta{V}=|V_{n}|-|V_{n+N}|\approx\sum_{i=n+1}^{n+N}I_{i-1,i}\cdot \mathbf{Z_{(i-1,i),abc}}
\end{equation}
where, $|V_n|, |V_{n+N}|$ are the voltage magnitude measurements of the observable nodes, $I_{i-1,i}$ and $\mathbf{Z_{(i-1,i),abc}}$ are the branch current and the phase impedance matrix between bus $i-1$ and $i$. The above equation can be rewritten in terms of nodal power measurements, as follows \cite{KW2016}:
\begin{equation}
\label{eq:normal_real}
\Delta{V}\approx{\sum_{i=n+1}^{n+N}\sum_{j=i}^{n+N}K_{i-1,i}\cdot l_{i-1,i}\cdot \frac{P_j}{cos\phi_j}}
\end{equation}
where, $K_{i-1,i}\ [\frac{\% drop}{kVA \cdot mile}]$ and $l_{i-1,i}$ are the approximate voltage drop factor and the length of distribution line segment between nodes $i-1$ and $i$, $P_j$ and $\cos\phi_j$ represent the nodal power consumption and the power factor at node $j$. When outage happens at an unobservable node $s$ downstream of node $n$, $n\leq s \leq n+L$, the post-outage voltage drop value, $\Delta{V_{o}}$, is determined as follows:
\begin{equation}
\label{eq:outage_real}
\Delta{V_{o}}\approx{\Delta{V}+\sum_{i=n}^{s-1}K_{i-1,i}\cdot l_{i-1,i}\cdot\frac{\Delta P_s}{cos\phi_s}}
\end{equation}
where, $\Delta P_s$ represents the outage event magnitude and has a negative value. Comparing \eqref{eq:outage_real} with \eqref{eq:normal_real}, we can observe that the voltage drop value across the two observable nodes changes after an outage event downstream of any of the two nodes. These changes are almost proportional to the outage magnitude, $\Delta{P_s}$. This can also be confirmed using real AMI data, as shown in Fig \ref{fig:3d}. This figure shows the perceivable gap between the joint data distribution obtained from two observable nodes under normal and outage conditions, in three dimensions. Given that an outage event anywhere downstream of the two nodes will lead to deviations from their underlying joint measurement data distribution in normal operations, we define an outage detection zone as follows:
\begin{definition}\label{def}
    In a radial network, an outage detection zone, $\Psi_i$, is defined as $\Psi_i = \{\omega_1,\omega_2,Z_{\Psi_i}\}$ where $\omega_1$ and $\omega_2$ are two observable nodes, with $\omega_1$ being upstream of $\omega_2$, and $Z_{\Psi_i}$ is the set of all the branches downstream of $\omega_1$.
\end{definition}

\subsection{Zone Selection}
\begin{figure}[tbp]
      \centering
      \includegraphics[width=3.3in]{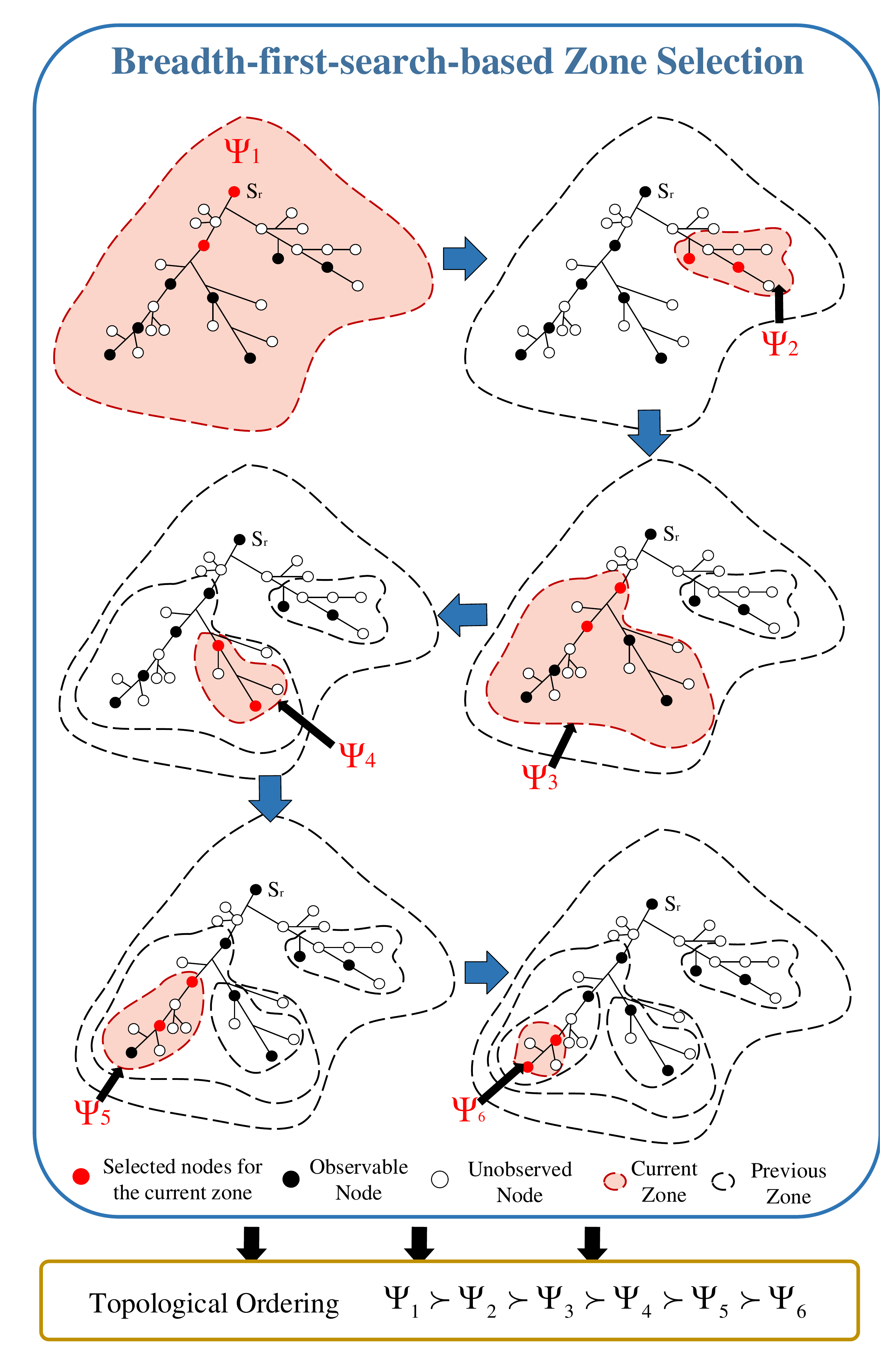}
\caption{Proposed BFS-based zone selection and ordering method.}
\label{fig:zone}
	\vspace{-1.5em}
\end{figure}

Based on Definition \ref{def}, each network can have different sets of zones based on the choice of observable nodes. In this paper, we propose a BFS-based zone selection method by exploiting the tree-like structure of distribution systems. As will be elaborated in Section \ref{zone}, the proposed zone selection algorithm offers two advantages: (1) it is able to obtain the \textit{optimal zone set}, which maximizes the outage location information in any partially observable network. (2) The proposed BFS-based algorithm introduces a \textit{valid topological ordering}, which significantly simplifies outage location identification process. The proposed algorithm involves the following steps: 

\begin{itemize}
\item \textbf{Step I:} Consider a partially observable distribution system, $g$, with a total number of $M$ branches, $B_g = \{b_1,...,b_M\}$, and a set of $O+1$ observable nodes, $S_g = \{S_r,S_1,S_2,...,S_O\}$, where $S_r$ represents the network's root node (i.e. main substation). 
\item \textbf{Step II:} Define and initialize the zone set for $g$, as $\Psi^g = \{\emptyset\}$. Note that the set $\Psi^g$ is an \textit{ordered set}, where new elements are added to the right side of the current elements in the set (i.e. order of elements matters). Initialize the set of candidate observable nodes as $S_B = \{S_r\}$, and the zone counter $k \gets 1$.
\item \textbf{Step III:} Select a node, $S_{o1}$, randomly from $S_B$. Remove $S_{o1}$ from $S_B$. Find all immediate observable nodes downstream of $S_{o1}$, denoted as $S_N$, and add them to $S_B$. Randomly select a node from the set $S_N$, denoted as $S_{o2}$.
\item \textbf{Step IV:} Select a new zone $\Psi_{k}$, with $\omega_1 \leftarrow S_{o1}$, $\omega_2 \leftarrow S_{o2}$, and include all the branches downstream of $S_{o1}$ into $Z_{\Psi_k}$ (see Definition I). Add $\Psi_k$ to the right side of the current zones in $\Psi^g$. 
\item \textbf{Step V:} $k \gets k+1$. Go back to Step II until $S_N$ is empty for all the nodes in $S_B$, as shown in Fig. \ref{fig:zone}. 
\item \textbf{Step VI:} Output the ordered set of all network zones, $\Psi^g = \{\Psi_1,...,\Psi_w\}$, with $w$ denoting the number of selected zones.
\end{itemize}

Following the proposed zone selection method, each branch in the system will belong to at least one zone, while at the same time, no two zones have the exact same set of branches. For example, branches of the zone $\Psi_6$ in Fig. \ref{fig:zone}, are also covered by zones $\Psi_1,...,\Psi_5$. As will be shown in Section \ref{zone}, these inter-zonal intersections introduce a \textit{redundancy}, which will be leveraged for enhancing the robustness of the outage detection process by blocking bad data samples and outliers. Furthermore, to specify the outage location considering the zonal intersections, a zone coordination method is proposed in Section \ref{GAN}.


\begin{figure*}[tbp]
      \centering
      \includegraphics[width=2\columnwidth]{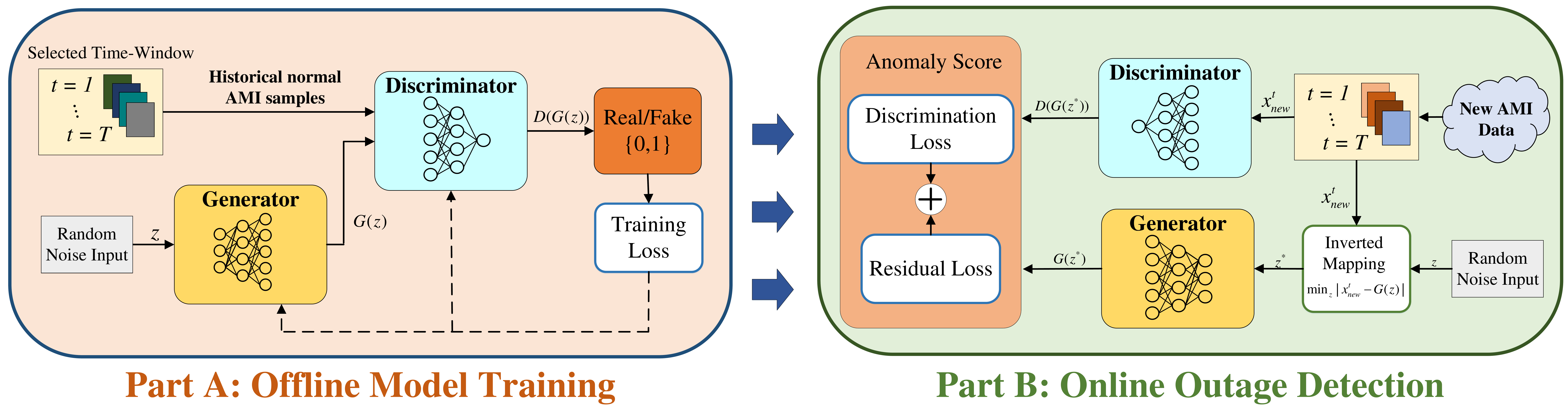}
\caption{GAN-based learning and testing structure.}
\label{fig:overall}
	\vspace{-1em}
\end{figure*}

\section{GAN-Based Zone Monitoring}\label{GAN}
\begin{algorithm}
\caption{GAN Training for zone $\Psi_i$}\label{alg:GAN}
\begin{algorithmic}[1]
\Require: {Seasonal normal behavior data for zone $\Psi_i$}
\Require: {Learning rate $\alpha$, batch size $m$, number of iterations for $D$ per $G$ iteration $n_D$, initial learning parameters for $G$ and $D$, $\theta_D$ and $\theta_G$}
    \While {Nash equilibrium has not been achieved}
        \For {$t = 0,...,n_D$}
            \State {Generate sample batch from the latent space $z$}
            \State {$p_z \to \{(z_j)\}^m_{j=1}$}
            \State {Obtain sample batch from the historical data}
            \State {$p_{X_{\Psi_i}} \to \{x_{\Psi_i}(j)\}^m_{j=1}$}
            \State {Update discriminator parameters using gradient descent with $\alpha$ based on the discriminator loss}
            \State {$\delta_D = \frac{1}{m}\sum_{j=1}^m[-\log{D(x_{\Psi_i}(j))}-\log(1-D(G(z_i)))]$}
            \State {$\theta_D := \theta_D - \alpha * \bigtriangledown_{\theta_D}\delta_D$}
        \EndFor    
        \State {Update generator parameters using gradient descent with $\alpha$}
        \State {$\delta_G = \frac{1}{m}\sum_{j=1}^m[-\log{D(G(z_j))}]$}
        \State {$\theta_G := \theta_G - \alpha * \bigtriangledown_{\theta_G}\delta_G$}
    \EndWhile
\end{algorithmic}
\end{algorithm}
In this paper, to quantify deviations from the measurement data distribution in normal conditions, we have utilized a recently-invented non-parametric unsupervised learning approach, GAN, which is able to implicitly represent complex data distributions without constructing high-dimensional likelihood functions \cite{GAN2016}. This addresses the challenge of dimensionality. Also, GAN does not assume a prior parametric structure over the data distribution. This ensures the performance of GAN for outage detection problem, since the utilities generally do not have a priori knowledge of the exact structure of data distribution in normal conditions. Meanwhile, since model training is done using only the data from normal condition, GAN is not vulnerable to the outage data scarcity problem. When training is completed, the data distributions of the zones in normal condition are represented by DNNs. Then, a GAN-based anomaly score is assigned to real-time measurements to detect outage events inside the zone \cite{GAN2017}.  

\subsection{GAN Fundamentals and Training Process}

For each zone, a GAN is trained to learn the joint distribution of measured variables $X = \{\Delta V^t,P_n^t,P_{n+N}^t\}_{t=1}^{T}$ within a time-window with length $T$ (see Fig. \ref{fig:outage_inside}), where $P_n^t$ and $P_{n+N}^t$ are the nodal power consumption for the two observable nodes in the zone, and $\Delta V^t$ is the voltage difference between the two nodes at time $t$. The purpose of defining a time-window over the observable variables is to exploit temporal relations between consecutive data samples in power distribution systems for more effective anomaly detection. The training set consists of the SM data history of the variables defined in each zone, and is denoted as $X_{\Psi_i}$ for zone $\Psi_i$. To account for the strong seasonal changes in customers' behavior that might mislead detecting the boundary between normal and outage behavior \cite{kaveh2019}, the dataset has been decomposed into separate seasons to train different GAN models for each zone. Each dataset is randomly divided into three separate subsets for training (70$\%$ of the total data), validation (15$\%$ of the total data), and testing (15$\%$ of the total data).

GAN relies on two interconnected DNNs, which are simultaneously trained via an adversarial process: a \textit{generator}, $G$, and a \textit{discriminator}, $D$ \cite{GAN2015}, as shown in Fig. \ref{fig:overall} (part A). The interaction between the two DNNs can be modeled as a game-theoretic two-player nested minmax optimization \cite{GAN2014}:
\begin{equation}
\begin{split}
\min \limits_{\theta_G} \max \limits_{\theta_D} V(D,G) & = \mathbb{E}_{x_{\Psi_i}\sim p_{X_{\Psi_i}}(x_{\Psi_i})}[\log(D(x_{\Psi_i}))] \\
&+\mathbb{E}_{z \sim p_z(z)}[\log(1-D(G(z)))] 
\end{split}
\end{equation}
where, $\theta_G$ and $\theta_D$ are the learning parameters of $G$ and $D$, respectively. $p_{X_{\Psi_i}}$ is the underlying probability density function of historical data obtained from the two observable nodes of the zone. In each iteration, $D$ is trained to maximize the probability of assigning the correct label to both training examples and artificially generated samples from $G$. Thus, the output of $D$, $0\leq D(x_{\Psi_i}) \leq 1$, represents the probability that $x_{\Psi_i}$ is from the training dataset rather than generated artificially by $G$ \cite{GAN2014}. On the other hand, $G$ is trained to generate artificial samples that maximize the probability of the discriminator $D$ mislabeling. The input of $G$ is defined as $z$, which is a noise signal with uniform distribution $p_{z}(z)$. A minibatch stochastic gradient descent is applied for training the GAN by updating the G-D model parameters cooperatively. After a number of training iterations, $G$ and $D$ will reach a unique global optima at which both cannot improve. This means the generator can recover the underlying distribution of the training data and the discriminator cannot distinguish the true samples from the artificially generated samples \cite{YZC2018}. The training process takes place offline and the detailed procedure is presented in Algorithm \ref{alg:GAN}.

\subsection{GAN-based Anomaly Score Assignment}
To detect potential outage events in each zone, a GAN-based anomaly score is utilized to evaluate sequential measurements of SMs online \cite{GAN2017}, as shown in Fig. \ref{fig:overall} (part B). The anomaly score consists of two loss metrics: the residual loss $\delta_R(\cdot)$ and the discriminator loss $\delta_D(\cdot)$. When a new data inquiry $x_{new}^t$ is obtained at time $t$, the residual loss describes the extent to which $x_{new}^t$ follows the learned distribution of the $G$ model, in the best case \cite{GAN2017}:
\begin{equation}\label{Rloss}
\delta_R(x_{new}^t) = \min_{z}|x_{new}^t-G(z)|
\end{equation}
After training, the generator, $G$, has learned an almost perfect mapping from the latent space $z$ to the zonal measurement data distribution in normal conditions. Hence, if $x_{new}^t$ is obtained from normal conditions, its residual loss value is zero, $\delta_D(x_{new}^t) = 0$, since $x_{new}^t$ and $G(z^*)$ are identical, where $z^*$ is the optimal solution to \eqref{Rloss}. Thus, higher $\delta_R(x_{new}^t)$ values represent deviations from normal operation conditions, suggesting occurrence of outage event within the zone.

The discriminator loss, $\delta_D(x_{new})$, is defined using the trained discriminator, $D$, to measure how well $G(z^*)$ follows the learned data distribution by the $G$ model. The discriminator loss can be written as \cite{GAN2014}:
\begin{equation}
\delta_D(x_{new}^t)=-\log{D(x_{new}^t)}-\log(1-D(G(z^*)))
\end{equation}
The GAN-based anomaly score for zone $\Psi_i$ is defined as the weighted sum of both loss metrics \cite{GAN2017}:
\begin{equation}\label{score}
\zeta_{\Psi_i}(x_{new}^t)=(1-\lambda)*\delta_R(x_{new}^t)+\lambda*\delta_D(x_{new}^t)
\end{equation}
where, $0\leq\lambda\leq1$ is a user-defined weight factor, the value of which is set at $\lambda = 0.1$ in this paper, based on suggestions in the literature \cite{GAN2017}. To determine the critical threshold for the anomaly score, above which new data points are identified as outage events, the GAN-based anomaly score, $\zeta_{\Psi_i}$, is obtained for all training data samples of zone $\Psi_i$. The sample mean, $\mu_{\Psi_i}$ and the sample variance, $\sigma_{\Psi_i}$, of the anomaly scores for the training data samples are calculated to determine the range of anomaly score in normal operations. When outage occurs, the real-time measurement data samples are expected to have anomaly scores above this range. The details of anomaly identification process are elaborated in the next section.

\subsection{GAN-based Zone Coordination}
Using the trained GANs, outage events can be detected in each zone by comparing the anomaly scores between the new inquiry samples and the critical threshold. Considering that a zone consists of a number of branches, a high anomaly score simply implies outage somewhere in the zone. To accurately pinpoint outage location in a large-scale distribution system, it is necessary to coordinate and combine anomaly scores from multiple zones. To achieve this, the following steps are performed:  

\begin{itemize}
\item \textbf{Stage I:} Assign a GAN to each zone, $\Psi_i\in\Psi^g$ and use Algorithm \ref{alg:GAN} over the historical seasonal data of the two observable nodes of each zone to learn the joint distribution of the measurement data. 
\item \textbf{Stage II:} After training for each zone, $\Psi_{i}$, obtain the anomaly score for training samples in the zone; determine the anomaly score sample mean and sample variance, denoted as $\mu_{\Psi_{i}}$ and $\sigma_{\Psi_{i}}$, respectively.
\item \textbf{Stage III:} At time $T$, observe the anomaly scores of all the zones in the set $\Psi^g$ based on the latest real-time measurements.
\item \textbf{Stage IV:} Select the first zone from the right side of the set $\Psi^g$ that has an abnormal anomaly score value and denote it as $\Psi_a$. We will show that this zone contains the maximum information on the outage event in Section \ref{zone}. In other words, $a = \argmaxA_{\xi}\xi,\ s.t.\ \zeta_{\Psi_\xi}>\mu_{\Psi_{\xi}} + h\cdot\sigma_{\Psi_{\xi}}$, where, $h$ is a user-defined threshold factor.
\item \textbf{Stage V:} Output the set of candidate branches that are potentially the location of outage event as $B_c = \Psi_a \setminus \{\Psi_{a+1}\cup \Psi_{a+2}\cup ...\cup \Psi_w\}$, where $A\setminus B$ represents the elements of set $A$ that are not in set $B$.
\end{itemize}

Based on the outcome of zone coordination, DSO can obtain the minimum branch candidates that are potentially impacted by the outage. This process will help the repair crew to rapidly find the outage location. Note that given the unbalanced nature of distribution networks, the proposed algorithm is applied to each phase separately. Hence, in practice, the zone set needs to be obtained for three phases. For the sake of conciseness we will continue our discussions for one phase, keeping in mind that the same logic applies to the other phases as well.

\section{Theoretical Properties of the Proposed Framework}\label{zone}
In this section, we discuss the theoretical properties of the proposed outage-detection framework. We will show that this approach has three fundamental properties: 

\textit{Framework Property 1 - Valid Topological Ordering of the Zones:} The framework introduces a \textit{valid topological order} among the zones to simplify the outage location process for large-scale networks. A valid topological order for any pair of zones is a relationship denoted as $\Psi_i \succ \Psi_j$, indicating that $\Psi_i$ has a higher topological order than $\Psi_j$. This means that $\Psi_i\not\subset \Psi_j$; i.e. either all branches in $\Psi_j$ are located downstream of the branches of $\Psi_i$ or the branches of $\Psi_i$ and $\Psi_j$ do not share any common path starting from the network's root node. Note that $\Psi^g = \{\Psi_1, ... , \Psi_w\}$ obtained from the proposed BFS-based zone selection algorithm follows a valid topological order, meaning that $\Psi_1\succ...\succ \Psi_w$. The reason for this is that the proposed zone selection algorithm explores all the immediate downstream nodes at the each depth level without backtracking in Stage II (Section \ref{framework}), prior to moving to the next level. 

To show this, note that when an outage event happens the anomaly scores for a subset of zones, ${\Psi^g}$, will increase above their normal range, where due to the radial structure of the networks these zones will follow a relationship of the form $\Psi_1\supset \Psi_2\supset ... \supset \Psi_{v_O}$, with $v_O$ denoting the number of the zones containing the faulted branch. Thus, the zones within $\Psi_g$ that are impacted by outage also follow a valid topological order. At Stage IV (Section \ref{GAN}), the proposed zone coordination algorithm selects $\Psi_{v_o}\leftarrow \Psi_{a}$ (i.e. the zone with the lowest topological order) as the zone that has the most specific information on the location of outage among all the impacted zones, since it contains the least number of candidate branches. Hence, higher order zones on the same path with abnormal anomaly scores, which are supersets of the selected zone and have less information on outage location, are automatically ignored. This eliminates the need for a burdensome comprehensive search process. Finally, to infer the candidate branches that are potentially the location of the outage event, all the branches in the healthy zones with lower topological orders than $\Psi_{v_O}$ have to be removed, as shown in Step IV (Section \ref{framework}). This helps the operator to directly pick the smallest set of branches among thousands of candidate branches in a large-scale network. For example, when outage occurs in any branches within $\Psi_6$ in Fig. \ref{fig:zone}, the DSO can ignore the anomaly scores of zones that have a higher topological ordering (i.e. $\Psi_1,...,\Psi_5$) to directly infer outage location as $\Psi_a \leftarrow \Psi_6$.

\textit{Framework Property 2 - Maximum Outage Location Information Extraction:} The proposed algorithm is able to obtain the optimal zone set as it maximizes the amount of information on the location of outage events in partially observable systems. To show this, first, we leverage the concept of entropy to assess the amount of outage location information in $\Psi^g$. The set $\gamma^g(b_j)$ is defined as $\gamma^g(b_j) = \{\forall \Psi_i:\ b_j \in \Psi_i\, \forall \Psi_i\in\Psi^g\}$. Hence, $\gamma^g(b_j)$ is the set of all zones in $\Psi^g$ that include $b_j$. Based on this definition, for each $\Psi^g$, a set of \textit{undetectable branch sets} is defined as $U(\Psi^g) = \{u_1,...,u_V\}$, where $u_k = \{b_{k_1},...,b_{k_n}: \forall b_{k_i}, b_{k_j}, \gamma^g(b_{k_i}) = \gamma^g(b_{k_j})\}$. Thus, $u_k$ defines a set of branches that are covered with the exact same set of zones and cannot be distinguished from each other in terms of outage event location. Given the set $U(\Psi^g)$ the outage location information can be measured using the concept of \textit{entropy}, as follows \cite{entropy2006}:
\begin{equation}
\label{eq:entropy}
H(U(\Psi^g))=-\sum_{i=1}^V \frac{|u_i|}{M} \log\frac{|u_i|}{M}
\end{equation}
where $|u_i|$ is the cardinality of the set $u_i$. The higher entropy value implies a higher number of distinguishable branches, and consequently, more information on outage location. The theoretical upper boundary for the entropy is $log(M)$; this case only happens when each $u_k$ only includes a single branch and $V = M$ (i.e. all branches are fully distinguishable and $|u_i| = 1$). This indicates any individual branch is distinguishable using two zones that intersect exactly at that branch. The theoretical lower boundary value for the entropy is zero, which implies that all the branches are covered by identical set of zones (i.e. the branches are not distinguishable and $|u_i| = M$). Based on this metric, the following theorem and proof are obtained:
\begin{figure}[tbp]
	\centering
	\includegraphics[width=3.5in]{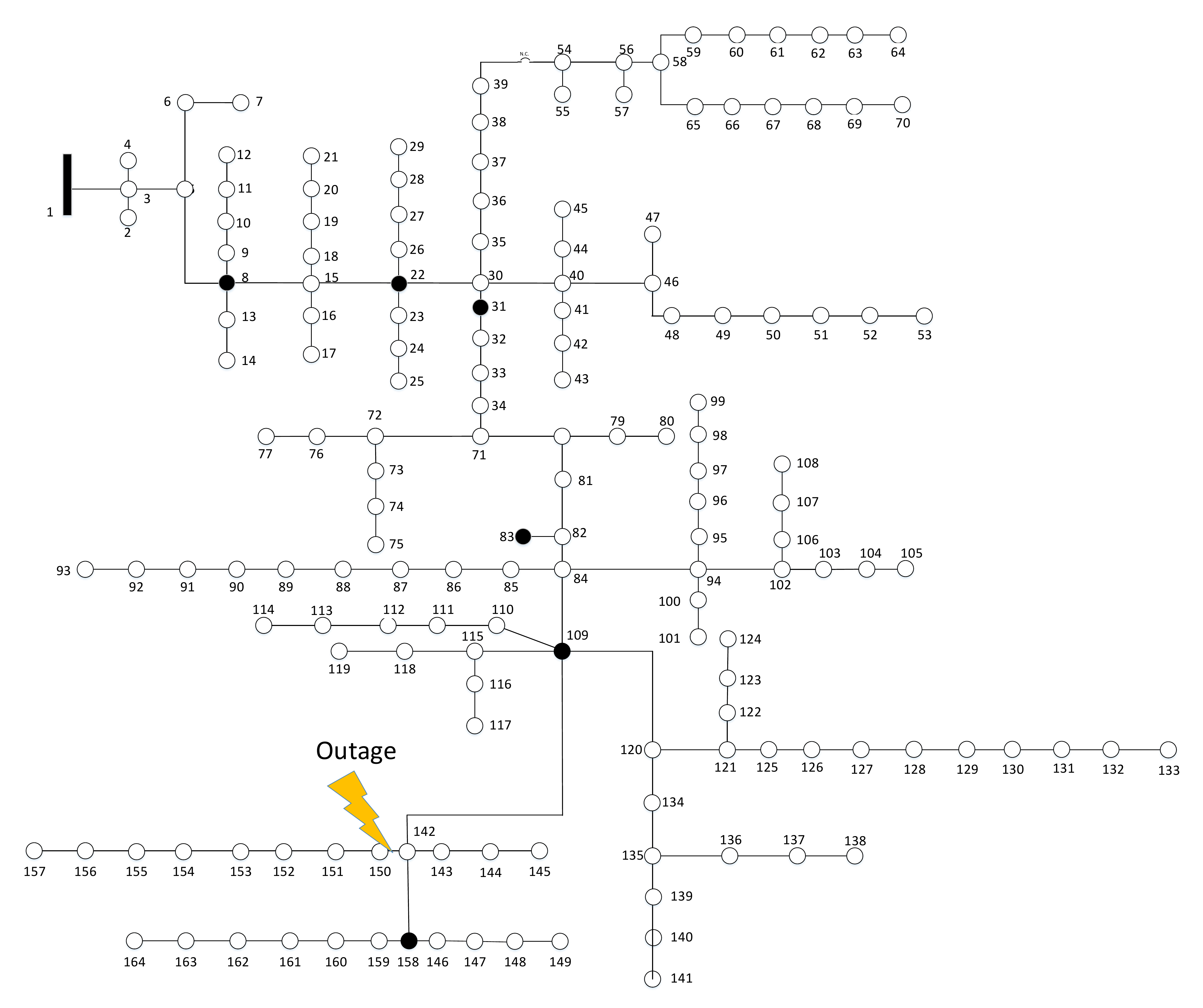}
	\caption{164-node feeder topology.}
	\label{fig:case}
	\vspace{-1em}
\end{figure}
\begin{figure}[tbp]
	\centering
	\includegraphics[width=3.5in]{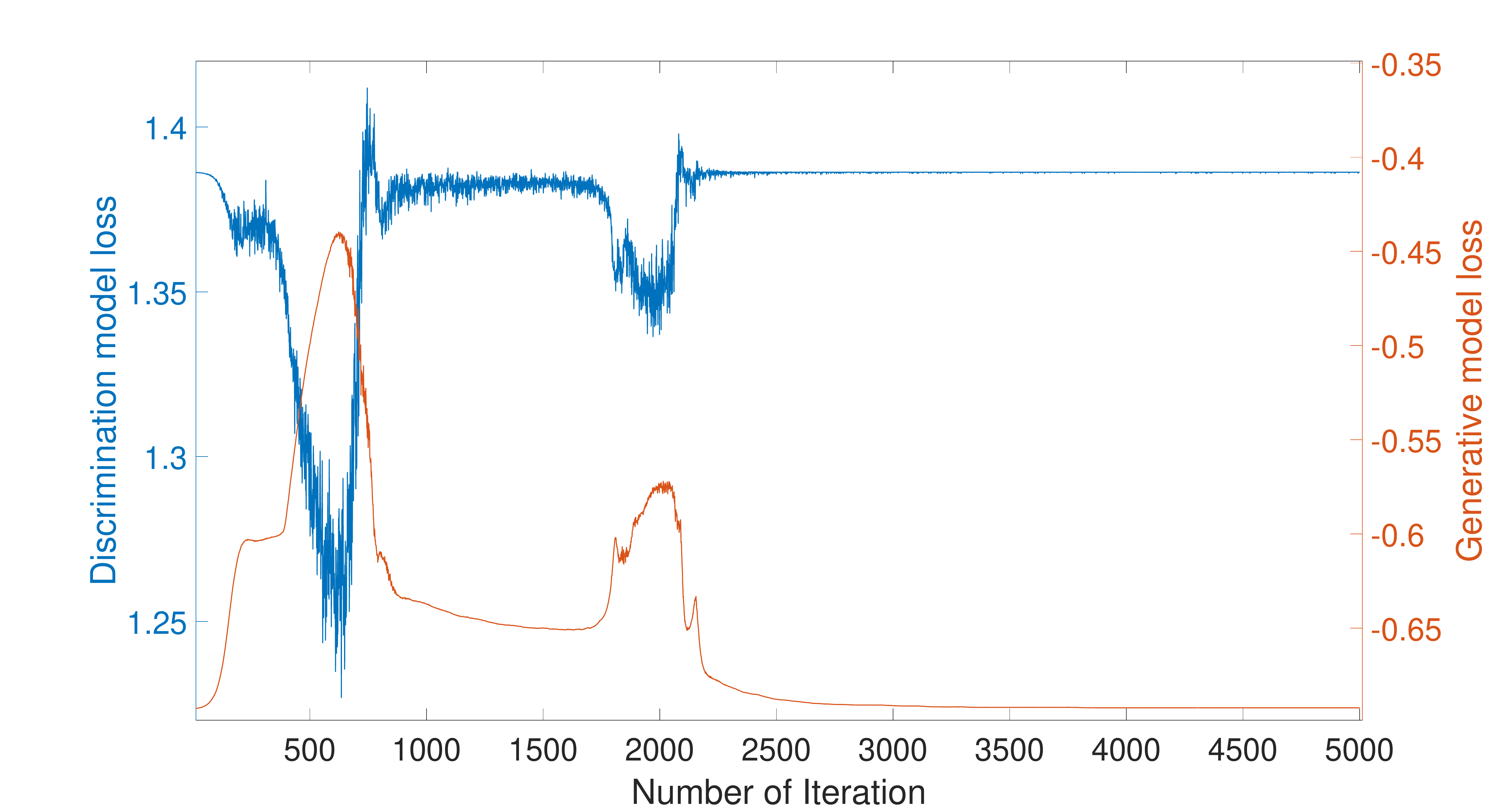}
	\caption{Training result for a GAN model.}
	\label{fig:converge}
	\vspace{-1em}
\end{figure}
\begin{theorem}\label{th_1}
For any partially observable network, the proposed BFS-based zone selection algorithm can find the optimal zone set that maximizes the outage detection entropy. 
\end{theorem}

\begin{proof}
To prove the theorem, we will show that addition or removal of a zone to $\Psi^g$ cannot increase the entropy. Thus, deviation from $\Psi^g$ cannot obtain additional outage location information. First, consider the case of removing an arbitrary zone $\Psi_j\in \Psi^g$, and without loss of generality assume that $\Psi_{j-1}\in \Psi^g$ and $\Psi_{j+1}\in \Psi^g$ are the smallest and largest zones, respectively, where $\Psi_{j-1}\supset \Psi_j \supset \Psi_{j+1}$ holds. Here, two undetectable branch sets can be identified: $u_{j-1} = \Psi_{j-1}\setminus \Psi_j$ and $u_j = \Psi_j\setminus \Psi_{j+1}$. Note that $\Psi_j$ is the only zone that enables discrimination between branches $u_j$ and $u_{j-1}$. Hence, if $\Psi_j$ is removed, $u_j$ will be eliminated from $U(\Psi^g)$, and $u_{j-1} \gets u_{j-1} \cup u_j$. This leads to a decrease in entropy, $H(U(\Psi^g))$, equal to $\frac{1}{M}\log\frac{(|u_{j-1}| + |u_j|)^{|u_{j-1}| + |u_j|}}{|u_{j-1}|^{|u_{j-1}|}|u_{j}|^{|u_{j}|}}$. This decrease shows that removal of any zone in $\Psi^g$ will reduce the amount of outage location information. Now consider the case of adding a zone to $\Psi^g$: assume that the newly added zone, $\Psi_j$, is defined by two observable nodes $S_{o1}\in S_g$ and $S_{o2}\in S_g$; however, the proposed algorithm has already utilized all the observable nodes in $S_g$ as $S_{o1}$, shown in Step II (Section \ref{framework}); this means that there is at least one zone in $\Psi^g$ that is identical to $\Psi_j$. Hence, adding a zone to the set $\Psi^g$ will not change $U(\Psi^g)$ and the entropy remains unchanged.
\end{proof}
\begin{figure}[tbp]
	\centering
	\includegraphics[width=3.5in]{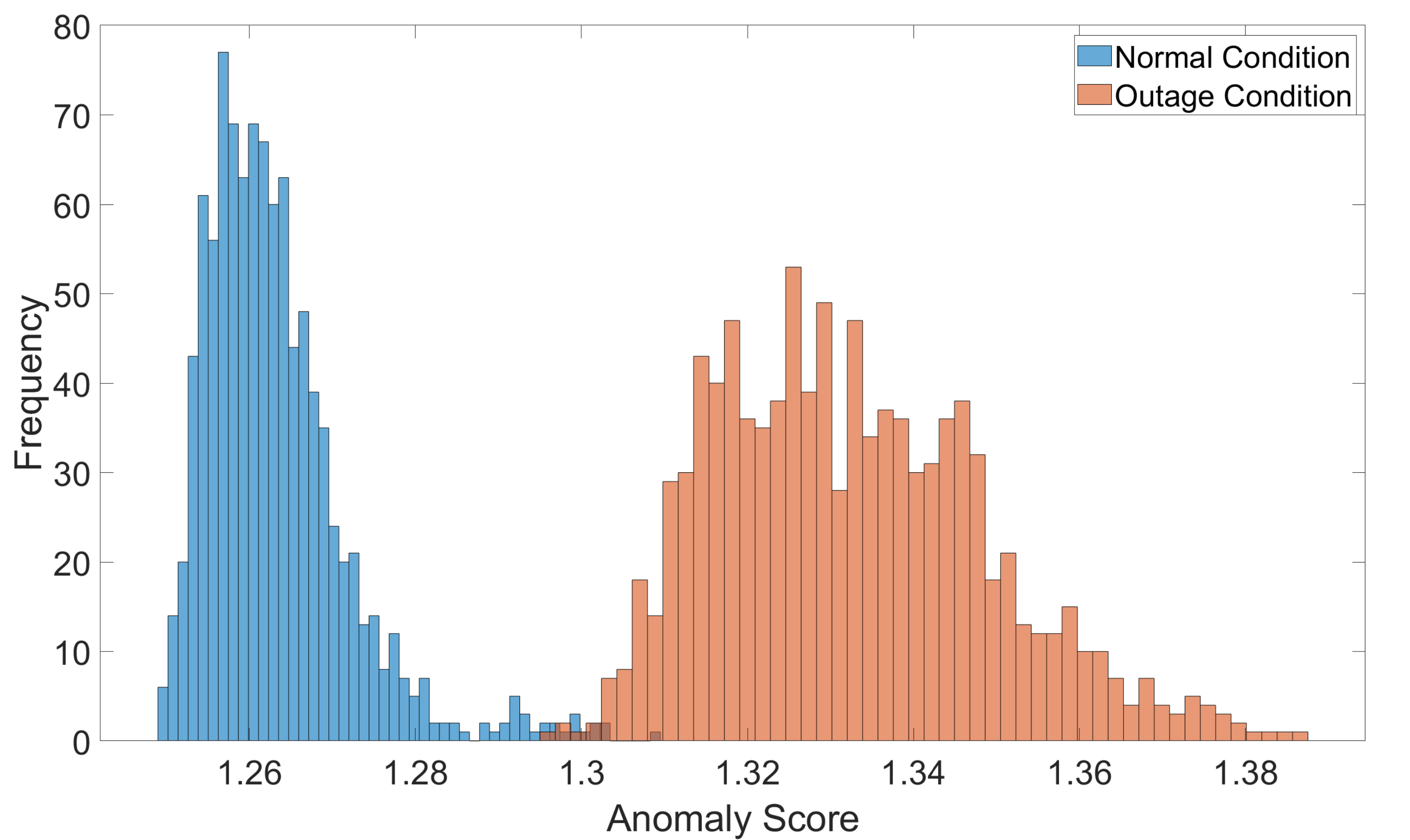}
	\caption{Anomaly score histogram under the normal and outage conditions.}
	\label{fig:case2}
	\vspace{-1em}	
\end{figure}
\begin{figure}[tbp]
	\centering
	\includegraphics[width=3.5in]{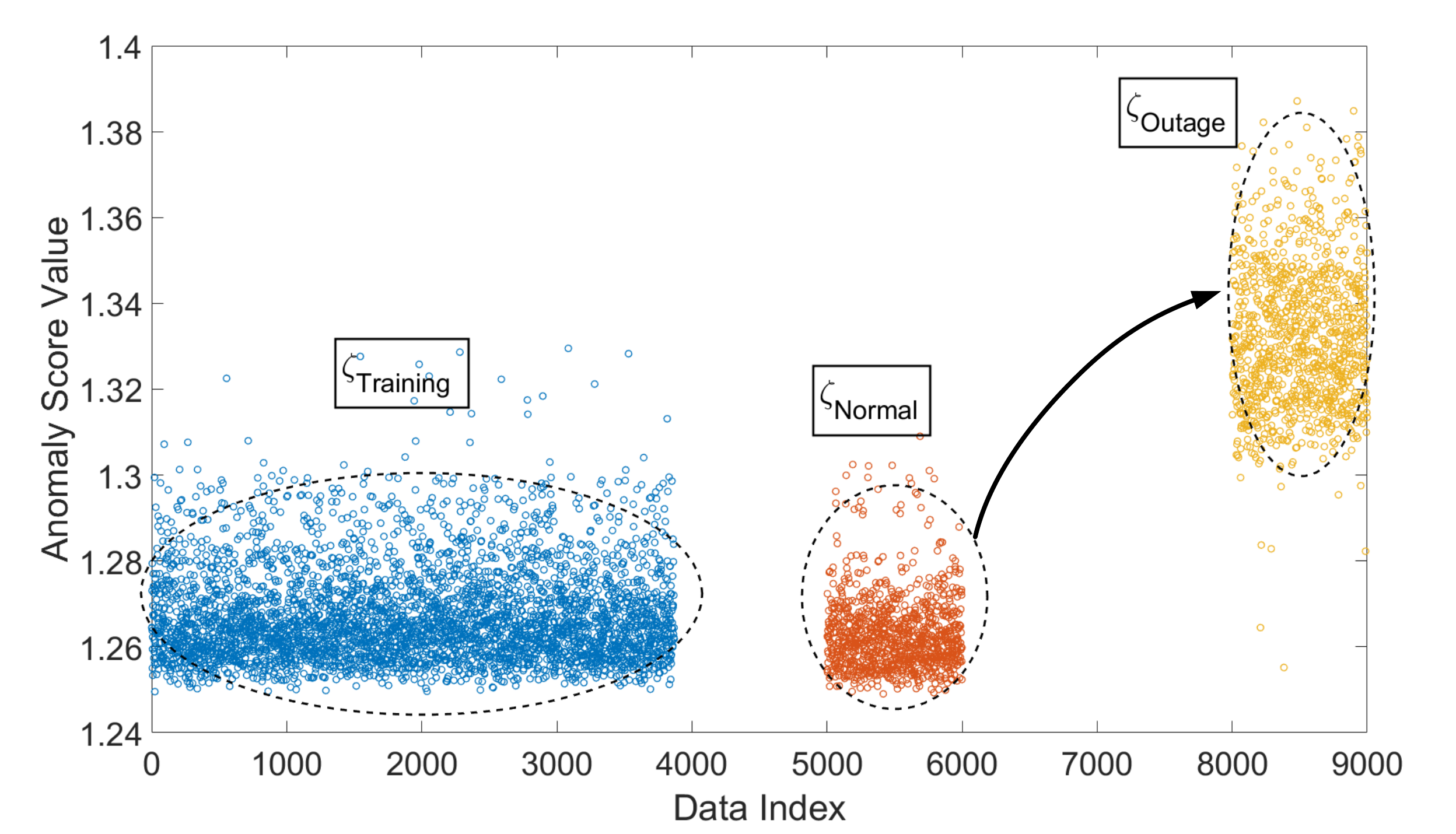}
	\caption{Anomaly score of the training set, with respect to the normal/outage test set.}
	\label{fig:case3}
	\vspace{-1em}	
\end{figure}
\textit{Framework Property 3 - Robustness Against Bad Data Samples:} Bad AMI data samples could generate high anomaly scores, which can lead to misclassification of bad data as outage event. Hence, it is essential to block these data samples from the outage detection algorithm. To do this, we have integrated a bad data detection mechanism into the algorithm by taking advantage of existing redundancy of the zones in $\Psi^g$. The basic idea is that since bad measurement data are not actually generated by outage events, it is highly unlikely to cause deviations in anomaly scores assigned to several intersecting zones at the same time, given that intersecting zones do not share the data from the same measurement devices. To introduce robustness against bad data, a set of redundant zones is selected for $\Psi_a$, Stage IV (Section \ref{GAN}). This set consists of the zones with lower topological order than $\Psi_a$, and is denoted as $\Psi^R = \{\Psi_{r_1},...,\Psi_{r_n}\}$, where $\Psi_a \subset \Psi_{r_i},\ \forall \Psi_{r_i}\in \Psi^R$. If $\exists \Psi_{r_i}$ such that $\zeta_{\Psi_{r_i}}\leq\mu_{\Psi_{r_i}} + h\cdot\sigma_{\Psi_{r_i}}$ then the outage in $\Psi_a$ is dismissed as bad data. The number of redundant zones $|\Psi^R|$ depends on the desired reliability of the algorithm against bad data. If the probability of receiving an anomaly due to bad data for each zone is $\eta$, then the probability of misclassifying a case of bad data as outage decreases with $\eta^{|\Psi^R|}$.   


\section{Numerical Results}\label{result}
The proposed outage detection method is tested on several real distribution feeders with AMI. The topology of one of these networks is shown in Fig. \ref{fig:case}. This 164-node feeder consists of residential (93$\%$) and commercial (7$\%$) customers. Six observable nodes are assumed in this feeder (node 8, node 22, node 31, node 83, node 109, and node 158), where five zones are defined based on these nodes. These zones are denoted $\{\Psi_1,...,\Psi_5\}$ and include branches downstream of node 8, node 22, node 31, node 83, and node 109, respectively. Note that $\Psi_1 \succ \Psi_2 \succ ... \succ \Psi_5$.

\begin{figure}[tbp]
	\centering
	\includegraphics[width=3.5in]{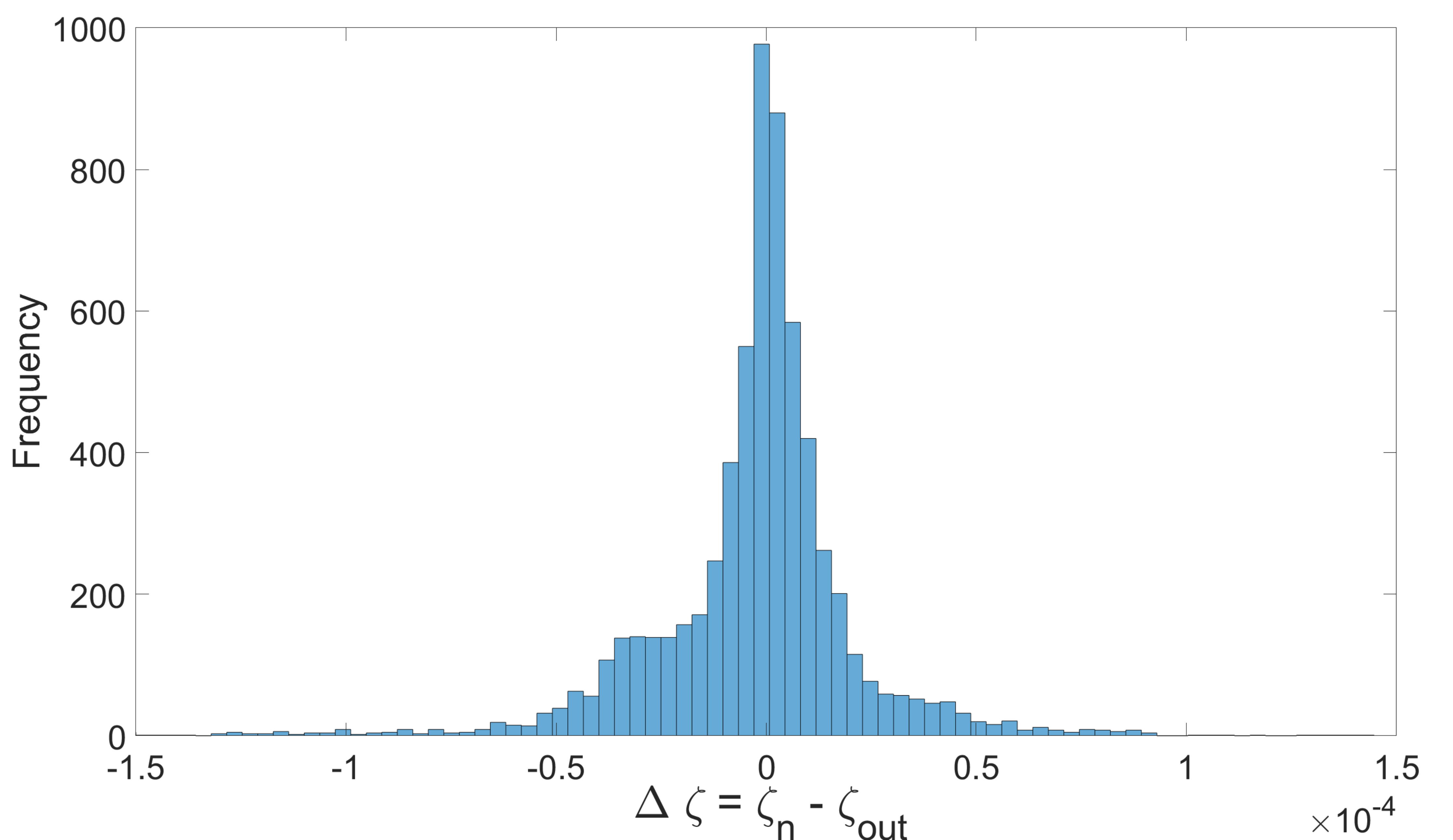}
	\caption{The histogram of $\Delta \zeta$.}
	\label{fig:case4}
	\vspace{-1em}
\end{figure}
\subsection{Performance of GAN Model}
To validate the performance of GAN training process, we calculate the loss values of $G$ and $D$ that states the model has converged to the Nash equilibrium or not. According to the theoretical analysis in \cite{GAN2014}, when the Jensen-Shannon divergence between the $G$ model's distribution and the data distribution is zero, the loss values of $G$ and $D$ should converge to $2\log(2)$ and $\log(\frac{1}{2})$ at the equilibrium, respectively. This has been confirmed in Fig. \ref{fig:converge}. After a number of training iterations, both $D$ and $G$ losses converge to the desired values and these indicate that the GAN has been trained successfully and the underlying joint data distribution in normal condition has been learned. 

\subsection{Performance of Outage Detection}
The performance of the GAN-based outage detection method is tested for different outage cases. The outage event is located between node $142$ and node $164$, as shown in Fig. \ref{fig:case}; three outage events are simulated with three different outage magnitudes to evaluate the performance of the proposed method. The first case is designed as a small-size event where around 20 customers are disconnected ($40 kW$ total average demand). The second case is designed to represent a middle-size event, where around 50 customers are impacted  ($100 kW$ total average demand).

The third case is a large-size event, with around 80 customers ($150 kW$ total average demand). For each case, GAN models are trained using the historical data of the five zones. Fig. \ref{fig:case2} presents the histogram of anomaly score for one zone under normal and outage conditions. The mean values of $\zeta$ are $1.263$ and $1.33$ in the normal and outage conditions with variance values $7.7\times10^{-5}$ and $2.7\times10^{-4}$, respectively. Based on Fig. \ref{fig:case2}, the difference between anomaly score under normal and outage conditions is large enough to enable DSOs to distinguish these conditions. Meanwhile, Fig. \ref{fig:case3} presents the consistency of anomaly score for training and test sets when the system is in normal conditions. However, when the outage event takes place in the zone, the real-time anomaly score reaches considerably higher values. 

It is critical to show that an outage event \textit{outside} a zone will not lead to abnormal anomaly scores for that zone. Fig. \ref{fig:case} shows the distribution of anomaly score changes for one zone, when the outages of different magnitudes happen outside the zone. Hence, this figure depicts the histogram of $\Delta \zeta = \zeta_{n} - \zeta_{out}$, where $\zeta_{n}$ is the anomaly score obtained in normal conditions and $\zeta_{out}$ is the anomaly score obtained when the outage happens outside the zone. As can be observed, the anomaly score assigned to the zone does not change and remains almost constant for these outside-zone outages, which indicates that the anomaly score can be relied upon to correctly distinguish the outages inside and outside the zone.

To evaluate the quality of outage detection performance of the proposed method for a multi-zone network, several statistical metrics are applied, such as accuracy (Accu), precision (Prec), recall, $F_1$ score, and the Area under the  Curve (AUC) \cite{roc2006}. The values of these indexes are presented in Table. \ref{table:1.1} for the three outage cases and different zones. Based on the results, we can conclude that the performance of the proposed outage detection method improves as the event size increases, due to higher levels of deviation from normal joint measurement data distribution. For medium and large outage cases, all indexes reach values over 0.9. Hence, based on this AMI dataset and the test feeders, the proposed method can accurately detect outage events in the partially observable systems.

\begin{table}
\centering
\setlength{\tabcolsep}{2.2mm}
\renewcommand\arraystretch{1.4}
\caption{Outage Detection Quality Analysis}
\begin{tabular}{ccccccc}
\hline\hline
Zone & Case & Accu & Recall & Prec & $F_1$ & AUC\\
\hline
\multirow{3}*{$\Psi_1$} & {\centering} case 1 & 0.752 & 0.645 & 0.8206  & 0.7223 & 0.7641\\
~ & case 2 & 0.913 & 0.967 & 0.8727 & 0.9175 & 0.9179 \\
~ & case 3 & 0.928 & 0.9970 & 0.8761 & 0.9326 & 0.9363\\
\hline
\multirow{3}*{$\Psi_2$} & case 1 & 0.8355 & 0.784 & 0.874 & 0.8266 & 0.8391\\
~ & case 2 & 0.9435 & 1 & 0.8985 & 0.9465 & 0.9492\\
~ & case 3 & 0.9435 & 1 & 0.8985 & 0.9465 & 0.9492\\
\hline
\multirow{3}*{$\Psi_3$} & case 1 & 0.673 & 0.506 & 0.7685 & 0.6074 & 0.6947\\
~ & case 2 & 0.912 & 0.984 & 0.8601 & 0.9179 & 0.9207\\
~ & case 3 & 0.914 & 0.988 & 0.8606 & 0.9199 & 0.9233\\
\hline
\multirow{3}*{$\Psi_4$} & case 1 & 0.9225 & 0.884 & 0.964 & 0.9223 & 0.9285 \\
~ & case 2 & 0.953 & 0.939 & 0.966 & 0.9523 & 0.9534\\
~ & case 3 & 0.981 & 0.995 & 0.968 & 0.9813 & 0.9812\\
\hline
\multirow{3}*{$\Psi_5$} & case 1 & 0.834 & 0.738 & 0.9134 & 0.8164 & 0.8468\\
~ & case 2 & 0.9605 & 0.991 & 0.934 & 0.9617 & 0.962\\
~ & case 3 & 0.965 & 1 & 0.9346 & 0.9662 & 0.9673\\
\hline\hline
\end{tabular}
\label{table:1.1}
\vspace{-1em}
\end{table}

\section{Conclusion}\label{conclusion}
In this paper, we have presented a new data-driven method to detect and locate outage events in partially observable grids using SM measurements. The proposed GAN-based approach is able to implicitly represent the distribution of data in normal conditions and determine potential outage events online. The developed multi-zone outage detection mechanism is based on an unsupervised learning approach, which can address several challenges in outage detection: 1) the poor observability of system caused by the limited number of SMs. 2) data imbalance problem caused by outage data scarcity. 3) the high-dimensionality of the data caused by the temporal-spatial relationship. Meanwhile, our proposed robust BFS-based zone selection and ordering mechanism is guaranteed to capture the maximum amount of information on outage location for any given partially observable system. This method is validated on a real utility feeder using real SM data.

\ifCLASSOPTIONcaptionsoff
  \newpage
\fi



\bibliographystyle{IEEEtran}
\bibliography{IEEEabrv,./bibtex/bib/IEEEexample}
\end{document}